\documentclass[11pt]{article}
\usepackage{amsthm,amsfonts,amsmath,amssymb,amscd, accents, mathrsfs, mathtools}
\usepackage{authblk}
\usepackage{hyperref}
\usepackage{multicol}
                                \usepackage{verbatim}
                                \usepackage{graphicx}

\swapnumbers
                              
\mathtoolsset{showonlyrefs=false}
                                
\theoremstyle{plain}

\usepackage{enumitem}
\numberwithin{equation}{section}
\newtheorem{theorem}{Theorem}[section]

\newtheorem{assumption}[theorem]{Assumption}
\newtheorem{lemma}[theorem]{Lemma}
\newtheorem{proposition}[theorem]{Proposition}

\newtheorem{definition}[theorem]{Definition}
\newtheorem{notation}[theorem]{Notation}

\theoremstyle{remark}

\numberwithin{equation}{section}

\newcommand{\bR}{{\mathbb R}}

\newcommand{\cH}{{\mathcal H}}

\newcommand{\cC}{{\mathcal C}}

\newcommand{\cI}{{\mathcal I}}

\newcommand{\cB}{{\mathcal B}}
\newcommand{\cE}{{\mathcal E}}

\newcommand{\ket}[1]{\left\vert #1\right\rangle}
\newcommand{\bra}[1]{\left\langle #1\right\vert}

\def\idty{{\mathchoice {\mathrm{1\mskip-4mu l}} {\mathrm{1\mskip-4mu l}} %
{\mathrm{1\mskip-4.5mu l}} {\mathrm{1\mskip-5mu l}}}}

\newcommand{\Tr}{\mathrm{Tr}}

\newcommand{\be}{\begin{equation}}
\newcommand{\ee}{\end{equation}}
\newcommand{\bea}{\begin{eqnarray}}
\newcommand{\eea}{\end{eqnarray}}
\newcommand{\beann}{\begin{eqnarray*}}
\newcommand{\eeann}{\end{eqnarray*}}

\usepackage{color}

\begin{document}

\title{Coherence as entropy increment for Tsallis and R\'enyi entropies}
\author{Anna Vershynina}
\affil{\small{Department of Mathematics, Philip Guthrie Hoffman Hall, University of Houston, 
3551 Cullen Blvd., Houston, TX 77204-3008, USA}}
\renewcommand\Authands{ and }
\renewcommand\Affilfont{\itshape\small}

\date{\today}

\maketitle

\begin{abstract} 
Relative entropy of coherence can be written as an entropy difference of the original state and the  incoherent state closest to it when measured by relative entropy. The natural question is, if we generalize this situation to Tsallis or R\'enyi entropies, would it define good coherence measures? In other words, we define a difference between Tsallis entropies of the original state and the incoherent state closest to it when measured by Tsallis relative entropy. Taking R\'enyi entropy instead of the Tsallis entropy, leads to the well-known distance-based R\'enyi coherence, which means this expression defined a good coherence measure. Interestingly, we show that Tsallis entropy does not generate even a genuine coherence monotone, unless it is under a very restrictive class of operations. Additionally, we provide continuity estimate for R\'enyi coherence. Furthermore, we present two coherence measures based on the closest incoherent state when measures by Tsallis or R\'enyi relative entropy.
\end{abstract}

\section{Introduction}
Quantum coherence describes the existence of quantum interference, and it is often used in thermodynamics \cite{A14, C15, L15}, transport theory \cite{RM09, WM13}, and quantum optics \cite{G63, SZ97}, among few applications. Recently, problems involving coherence included quantification of coherence \cite{BC14, PC16, RPL16, R16, SX15, YZ16}, distribution \cite{RPJ16}, entanglement \cite{CH16, SS15}, operational resource theory \cite{CG16, CH16, DBG15, WY16},  correlations \cite{HH18, MY16, TK16}, with only a few references mentioned in each. See \cite{SAP17} for a more detailed review.

The golden standard for any ``good" coherence measure is for it to satisfy four criteria presented in \cite{BC14}: vanishing on incoherent states; monotonicity under  incoherent operations; strong monotonicity under incoherent operations, and convexity. Alternatively, the last two properties can be substituted by an additivity for subspace independent states, which was shown in \cite{YZ16}. See Preliminaries for more details.

 A number of ways has been proposed as a coherence measure, but only a few satisfy all necessary criteria \cite{BC14, ZY18, Zetal17}. A broad class of coherence measures are defined as the minimal distance $D$ to the set of incoherent states $\cI$, as
 $$CD(\rho)=\min_{\delta\in \cI}D(\rho, \delta). $$
 Here ``distance" is understood in a rather broad term, more of a distinguishability measure. We discuss the properties it should satisfy in chapters below.
It was shown in \cite{BC14} that for a relative entropy there is a closed expression of a distance-based coherence:
\begin{equation}\label{eq:rel}
\min_{\delta\in \cI}S(\rho\|\delta)=S(\rho\|\Delta(\rho))=S(\Delta(\rho))-S(\rho)\ ,
\end{equation}
here $\Delta(\rho)$ is the dephased state in a pre-fixed basis, see Notation \ref{notation}.

Different set of incoherent operations generate other physically relevant coherence measures. The largest set one considers is the set of incoherent operations (IO)  \cite{BC14}, which have Kraus operators that each preserve the set of incoherent states (see Definition \ref{def:IO}). A smaller set is called genuine incoherent operations (GIO) \cite{DS16}, which act trivially on incoherent states, see Definition \ref{def:GIO}. See \cite{CG16-2} for a larger list of incoherent operations, and their comparison. For these types of incoherent operations one may look at similar properties as the ones presented in \cite{BC14}. Restricted to GIO, one would obtain a measure of genuine coherence when it is non-negative and monotone, or a coherence monotone when it is also strongly monotone under GIO.

Motivated by the last expression in (\ref{eq:rel}), similar expressions were considered in \cite{DH21} for Tsallis and R\'enyi entropies:
$$S^R_\alpha(\Delta(\rho))-S_\alpha^R(\rho)\ , $$
$$S_\alpha^T(\Delta(\rho))-S_\alpha^T(\rho)\ . $$
It was found that these expressions define genuine coherence monotones (definition will come later). They have advantage over distance-based measures by being the explicit expressions, easy to calculate. Moreover, they can be regarded as measurement-induced entropy increment related to the quantum thermodynamics \cite{KA16}.

In \cite{V22} the following generalized genuine coherence monotone was proposed:
$$\cC_f(\rho)=S_f(\Delta(\rho))-S_f(\rho)\ , $$
here $S_f(\rho)$ is a quasi entropy, which could be defined in two ways, one of which is $S_f(\rho)=-S_f(\rho\|| I)$. 

Here we show the operational meaning of this $f$-coherence, by showing that it is not possible to distill a higher coherence states from a lower coherence state via GIO, Theorem \ref{thm:distill}. To prove this result, we first show the continuity of $f$-coherence, Theorem \ref{thm:coh-cont}.  

If one looks at (\ref{eq:rel}) again, the last expression is the difference in entropies of the state $\rho$ and its {closest incoherent state} $\Delta(\rho)$, when measured by the relative entropy. So we ask a question, if we change the entropy and relative entropy in this expression to the Tsallis ones, would that generate a good coherence monotone/measure? Note that this change will change the closest incoherent state as well.
In other words, we investigate the properties of the following Tsallis coherence
$$CT_\alpha(\rho):=S_\alpha^T(\Delta_\alpha(\rho))-S^T_\alpha(\rho)\ , $$
here $\Delta_\alpha(\rho)$ is the closest incoherent state to $\rho$ when measured by Tsallis relative entropy, i.e.
$$S^T_\alpha(\rho\|\Delta_\alpha(\rho)):=\min_{\delta\in \mathcal{I}}S^T_\alpha(\rho\|\delta)\ . $$
The explicit form of $\Delta_\alpha$ is given in \cite{R16}, and it is the same for R\'enyi and Tsallis relative entropies. 

Surprisingly, taking R\'enyi entropies above leads to the well-known distance-based R\'enyi coherence:
$$CR_\alpha(\rho)=\min_{\delta\in\mathcal{I}}S_\alpha^R(\rho\|\delta)=S_\alpha^R(\rho\|\Delta_\alpha(\rho))=S^R_\alpha(\Delta_\alpha(\rho))-S^R_\alpha(\rho)\ . $$
We provide a continuity estimate for this R\'enyi coherence \ref{thm:cont-R}.

This means that the entropy increment for von Neumann entropy (with relative entropy) and R\'enyi entropy are good coherence measures, however, we show that a similar Tsallis entropy does not lead even to a good genuine coherence monotone. It is a coherence monotone under a very restrictive class of operations.

At the end, we propose two new coherence measures, inspired by the expression for the closest incoherent state when measured by the Tsallis or R\'enyi relative entropy.

\section{Preliminaries}
\subsection{Coherence}

Let $\cH$ be a $d$-dimensional Hilbert space. Let us fix an orthonormal basis $\cE=\{\ket{j}\}_{j=1}^d$ of vectors in $\cH$.
\begin{definition} A state $\delta$ is called {\it incoherent} if it can be represented as follows
$\delta=\sum_j \delta_j\ket{j}\bra{j}. $
\end{definition}

\begin{notation}\label{notation}
Denote the set of {\bf incoherent states} for a fixed basis $\cE=\{\ket{j}\}_j$ as $\cI=\{\rho=\sum_jp_j\ket{j}\bra{j}\}.$
A {\bf dephasing} operation in $\cE$ basis is the following map:
$$\Delta(\rho)=\sum_j \bra{j}\rho\ket{j} \ket{j}\bra{j}\ . $$
\end{notation}

\begin{definition}\label{def:IO} A CPTP map $\Phi$ with the following Kraus operators
$$\Phi(\rho)=\sum_n K_n \rho K_n^*\ , $$
is called {\bf the incoherent operation (IO)} or incoherent CPTP (ICPTP), when the Kraus operators satisfy
$$K_n \cI K_n^*\subset \cI,\ \text{for all }n \ , $$
besides the regular completeness relation $\sum_n K_n^*K_n=\idty$.
\end{definition}
Consider each $K_n$, in \cite{YXGS15} it was shown that condition $K_n \cI K_n^*\subset \cI$ implies that there exists at most one nonzero entry in every column of $K_n$. 

Any reasonable measure of coherence $\cC(\rho)$ should satisfy the following conditions
\begin{itemize}
\item (C1) $\cC(\rho)\geq 0$, and $\cC(\rho)=0$ if and only if $\rho\in\cI$;
\item (C2) Non-selective monotonicity under IO (monotonicity): for all IO $\Phi$ and all states $\rho$,
$$\cC(\rho)\geq \cC(\Phi(\rho))\ ; $$
\item (C3) Selective monotonicity under IO (strong monotonicity): for all IO $\Phi$ with Kraus operators $K_n$, and all states $\rho$,
$$\cC(\rho)\geq \sum_n p_n \cC(\rho_n)\ , $$
where $p_n$ and $\rho_n$ are the outcomes and post-measurement states
$$\rho_n=\frac{K_n\rho K_n^*}{p_n},\ \ p_n=\Tr K_n\rho K_n^*\ . $$
\item (C4) Convexity, 
$$\sum_n p_n \cC(\rho_n)\geq \cC\left(\sum_n p_n\rho_n\right)\ , $$
for any sets of states $\{\rho_n\}$ and any probability distribution $\{p_n\}$.
\end{itemize}
Conditions (C3) and (C4) together imply (C2) \cite{BC14}.

Alternatively, instead of the last two conditions, one can impose the following one 
\begin{itemize}
\item (C5) Additivity for subspace-independent states: For $p_1+p_2=1$, $p_1, p_2\geq 0$, and any two states $\rho_1$ and $\rho_2$, 
$$\cC(p_1\rho_1\oplus p_2\rho_2)=p_1\cC(\rho_1)+p_2\cC(\rho_2)\ .$$
\end{itemize}
In  \cite{YZ16} it was shown that (C3) and (C4) are equivalent to  (C5) condition.

These properties are parallel with the entanglement measure theory, where the average entanglement is not increased under the local operations and classical communication (LOCC). Notice that coherence measures that satisfy conditions (C3) and (C4) also satisfies condition (C2). 

In \cite{DS16} a class of incoherent operations was defined, called {genuinely incoherent operations (GIO)} as quantum operations that preserve all incoherent states.

\begin{definition}\label{def:GIO} An IO map $\Lambda$ is called a {\bf genuinely incoherent operation (GIO)} is for any incoherent state $\delta\in\cI$,
$$\Lambda(\delta)=\delta\ . $$
\end{definition}

Additionally, it was shown that an operation $\Lambda$ is GIO if and only if all Kraus representations of $\Lambda$ has all Kraus operators diagonal in a pre-fixed basis  \cite{DS16}. 

Conditions (C2), (C3) and (C4) can be restricted to GIO and obtain different classes of coherence measures.
\begin{definition}
In this case, a {\bf genuine coherence monotone} satisfies at least (C1) and (C2). And if a coherence measure fulfills conditions (C1), (C2), (C3) it is called {\bf measure of genuine coherence}. 
\end{definition}

A larger class than GIO, called SIO, was defined in \cite{WY16, YMGGV16}.
\begin{definition}\label{def:SIO}
An IO $\Lambda$ is called {\bf strictly incoherent operation (SIO)} if its Kraus representation operators commute with dephasing, i.e. for $\Lambda(\rho)=\sum_j K_j\rho K_j^*$, we have for any $j$,
$$K_j\Delta(\rho)K_j^*=\Delta(K_j\rho K_j^*) \ . $$
\end{definition}

Since Kraus operators of GIO are diagonal in $\cE$ basis, any GIO map is SIO as well, i.e. GIO $\subset$ SIO, \cite{DS16}. 

A class of operators generalizing SIO, called DIO, was introduced in \cite{CG16}.
\begin{definition}\label{def:DIO}
An IO $\Lambda$ is called {\bf dephasing-incoherent operation (DIO)} if it itself commute with dephasing operator, i.e.
$$\Lambda(\Delta(\rho))=\Delta(\Lambda(\rho)) \ . $$
\end{definition}
Thus, we have GIO $\subset$ SIO $\subset$ DIO.

One may consider an additional property, closely related to the entanglement theory:
\begin{itemize}
\item (C6) Uniqueness for pure states: for any pure state $\ket{\psi}$ coherence takes the form:
$$\cC(\psi)=S(\Delta(\psi))\ , $$
where $S$ is the von Neumann entropy and $\Delta$ is the dephasing operation defined as
$$\Delta(\rho)=\sum_j \bra{j}\rho\ket{j}\ket{j}\bra{j}\ . $$
\end{itemize}
However, for other coherence measures the von Neumann entropy in (C6) may change to another one, and the dephased state may also change to another free state.

\subsection{R\'enyi and Tsallis coherences}
As mentioned before, relative entropy of coherence can be defined using three expressions
\begin{equation}\label{RelativeCoh}
C(\rho)=\min_{\delta\in\cI}S(\rho\|\delta)=S(\rho\|\Delta(\rho))=S(\Delta(\rho))-S(\rho)\ .
\end{equation}
Let us point out that $\Delta(\rho)$ is the closest incoherent state to $\rho$ when measured by relative entropy, which was shown in \cite{BC14}.

Recall, that Tsallis entropy is defined as for $\alpha\in(0,2]$
$$S^T_\alpha(\rho)=\frac{1}{1-\alpha}\left[\Tr\rho^\alpha-1\right]\ , $$
Tsallis relative entropy is defined as
$$S_\alpha^T(\rho\|\delta)=\frac{1}{\alpha-1}\left[\Tr\left(\rho^\alpha\delta^{1-\alpha} \right)-1 \right] $$
R\'{e}nyi entropy is defined as for $\alpha\in(0,\infty)$
$$S_\alpha^R(\rho)=\frac{1}{1-\alpha}\log\Tr\rho^\alpha\ , $$
and R\'enyi relative entropy is defined as
$$S_\alpha^R(\rho\|\delta)=\frac{1}{\alpha-1}\log\Tr\left(\rho^\alpha\delta^{1-\alpha} \right)\ . $$

Motivated by different forms involved in the definition of relative entropy of coherence (\ref{RelativeCoh}), R\'enyi coherence has been defined as
\begin{align}
CR_\alpha^1 (\rho)&=\min_{\delta\in\cI}S_\alpha^R(\rho\|\delta)\label{CR-min}\ ,\\
CR_\alpha^2(\rho)&=S_\alpha^R(\Delta(\rho))-S_\alpha^R(\rho)\label{CR-diff}\ , \\
CR_\alpha^3(\rho)&=S_\alpha^R(\rho\|\Delta(\rho))\label{CR-rel}\ .
\end{align}
The first definition $CR^1_\alpha$ is a particular case of any distance-based coherence \cite{BC14}, and was separately discussed in \cite{SL17}. The second definition $CR^2_\alpha$ was introduced in \cite{DH21}. The third definition $CR^3_\alpha$ was introduced in \cite{CG16-2}.

Similarly, Tsallis coherence has been defined as
\begin{align}
CT_\alpha^1 (\rho)&=\min_{\delta\in\cI}S_\alpha^T(\rho\|\delta)\label{CT-min}\ ,\\
CT_\alpha^2(\rho)&=S_\alpha^T(\Delta(\rho))-S_\alpha^T(\rho)\label{CT-diff}\ .
\end{align}
The first definition $CT^1_\alpha$ is a particular case of any distance-based coherence \cite{BC14}. The second definition $CT^2_\alpha$ was introduced in \cite{DH21}. 

These definitions are all different, in particular, due to the fact that the closest incoherent state to a state $\rho$, when measured by either R\'enyi or Tsallis relative entropy, is not a state $\Delta(\rho)$. From \cite{CG16-2, R16} the closest incoherent state to a state $\rho$ for either R\'{e}nyi or Tsallis relative entropies is
\begin{equation}\label{eq:Delta-alpha}
\Delta_\alpha(\rho)=\frac{1}{N(\rho)}\sum_j\bra{j}\rho^\alpha\ket{j}^{1/\alpha}\ket{j}\bra{j}\ \in \cI\ , 
\end{equation}
where $N(\rho)=\sum_j\bra{j}\rho^\alpha\ket{j}^{1/\alpha}$.
The corresponding relative entropy becomes
\begin{equation}\label{eq:CT1}
CT_\alpha^1(\rho)=S_\alpha^T(\rho\|\Delta_\alpha(\rho))=\frac{1}{\alpha-1}\left[ N(\rho)^\alpha-1\right]\ , 
\end{equation}
and 
\begin{equation}\label{eq:CR1}
CR_\alpha^1(\rho)=S_\alpha^R(\rho\|\Delta_\alpha(\rho))=\frac{\alpha}{\alpha-1}\log N(\rho)\ .
\end{equation}

Interestingly enough difference-based Tsallis coherence when $\alpha=2$ is related to the distance-based coherence induced by the Hilbert-Schmidt distance \cite{DS16}
$$C_2^{HS}(\rho):=\min_{\delta\in\cI}\|\rho-\delta\|_2^2=S_2^T(\Delta(\rho))-S_2^T(\rho)\ , $$
where $\|\rho-\delta\|_2^2=\Tr (\rho-\delta)^2$.

\subsection{Generalized coherences}
Any proper distance $D(\rho, \sigma)$ between two quantum states, can induce a potential candidate for coherence. The distance-based coherence measure is defined as follows \cite{BC14}.

\begin{definition} $$CD(\rho):=\min_{\delta\in\cI}D(\rho, \delta)\ ,$$i.s. the minimal distance between the state $\rho$ and the set of incoherent states $\cI$ measured by the distance $D$.\end{definition}

\begin{itemize}
\item (C1) is satisfied whenever $D(\rho, \delta)=0$ iff $\rho=\delta$.
\item (C2) is satisfied whenever $D$ is contracting under CPTP maps, i.e. $D(\rho, \sigma)\geq D(\Phi(\rho), \Phi(\sigma))$.
\item (C4) is satisfied whenever $D$ is jointly convex.
\end{itemize}

Since the relative entropy, R\'enyi and Tsallis relative entropies satisfy all three above conditions for $\alpha\in[0,1)$, (C1), (C2), and (C4) are satisfied for $C(\rho)$, $CR^1_\alpha$, $CT^1_\alpha$.

Another generalization was considered in \cite{V22}, which is based on quasi-relative entropy.
\begin{definition}\label{def:qre}
For strictly positive bounded operators $A$ and $B$ acting on a finite-dimensional Hilbert space $\cH$, and for any {continuous} function $f: (0,\infty)\rightarrow \bR$, {\it the quasi-relative entropy} (or sometimes referred to as {\it the $f$-divergence}) is defined as 
$$S_f(A|| B)=\Tr(f(L_BR_A^{-1}){A})\ ,$$
where left and right multiplication operators are defined as $L_B(X)=BX$ and $R_A(X)=XA$. 
\end{definition}

Having the spectral decomposition of operators one can calculate the quasi-relative entropy explicitly \cite{HM17, V16}. Let $A$ and $B$ have the following spectral decomposition
\begin{equation}\label{eq:spectral}
A=\sum_j\lambda_j\ket{\phi_j}\bra{\phi_j}, \ \ B=\sum_k\mu_k\ket{\psi_k}\bra{\psi_k}\ .
\end{equation}
Here the sets $\{\ket{\phi_k}\bra{\psi_j}\}_{j,k}$, $\{\ket{\psi_k}\bra{\psi_j}\}_{j,k}$ form orthonormal bases of $\cB(\cH)$, the space of bounded linear operators. By \cite{V16}, the  quasi-relative entropy is calculated as follows
\begin{equation}\label{eq:formula}
S_f(A||B)=\sum_{j,k}\lambda_j f\left(\frac{\mu_k}{\lambda_j}\right)|\bra{\psi_k}\ket{\phi_j}|^2\ . 
\end{equation}

\begin{assumption}
To define $f$-coherence,  we  assume that the function $f$ is operator convex and operator monotone decreasing and $f(1)=0$.
\end{assumption}

 $f$-entropy was defined in two ways in \cite{V22}
\begin{align}
{S}^1_f(\rho):&=-S_f(\rho\|I)=-\sum_j \lambda_j f\left(\frac{1}{\lambda_j}\right)\label{eq:f-entropy-hat}\\
S_f^2(\rho):&=f(1/d)-S_f(\rho\|I/d)=f(1/d)-\sum_j \lambda_j f\left(\frac{1}{d\lambda_j}\right)\label{eq:f-entropy}\ ,
\end{align}
where $\{\lambda_j\}_j$ are the eigenvalues of $\rho$.

\begin{definition} For either $f$-entropy, the $f$-coherence is then defined as
\begin{equation}\label{eq:def-coh-f-hat}
{C}_f(\rho):={S}_f(\Delta(\rho))-{S}_f(\rho)\ .
\end{equation}
\end{definition}
If $\{\lambda_j\}$ are the eigenvalues of $\rho$, and the diagonal elements of $\rho$ in $\cE$ basis are $\chi_j=\bra{j}\rho\ket{j}$, then from  (\ref{eq:f-entropy-hat}), we have
\begin{align*}\label{eq:c_f-l-c-hat}
{C}^1_f(\rho)&=\sum_j \lambda_j f\left(\frac{1}{\lambda_j}\right)-\sum_j \chi_j f\left(\frac{1}{\chi_j}\right)\\
C_f^2(\rho)&=\sum_j \lambda_j f\left(\frac{1}{d\lambda_j}\right)-\sum_j \chi_j f\left(\frac{1}{d\chi_j}\right)\ ,
\end{align*}

Since $f(x)=-\log(x)$ is operator convex, coherence measure defined above coincides with the relative entropy of coherence (\ref{RelativeCoh}) \cite{BC14}:
\begin{align*}
C_{\log}(\rho)&=S_{\log}(\Delta(\rho))-S_{\log}(\rho)=S(\Delta(\rho))-S(\rho)=C(\rho)\ .
\end{align*}

The function $f(x)=\frac{1}{1-\alpha}(1-x^{1-\alpha})$ is operator convex for $\alpha\in(0,2)$. The coherence monotone then becomes the Tsallis relative entropy of coherence
\begin{equation*}
C^1_\alpha(\rho)=\frac{1}{1-\alpha} \left[\sum_j \chi_j^\alpha-\sum_j\lambda_j^\alpha \right]=CT^2_\alpha(\rho)\ .
\end{equation*}

\subsection{Properties}
Here we list which properties (C1)-(C5) are satisfied by which coherences and under which conditions. For R\'enyi and Tsallis entropies we do not consider a case when $\alpha=1$ and the entropies reduce to the relative entropy of coherence.\\

\begin{tabular}{|c|c|c|c|c|c|c|c|}
\hline
&(C1)&(C2) under&(C3) under&(C4)&(C5)\\
\hline
$CD$ &\checkmark&{IO}  \cite{BC14}&X&\checkmark&\\
\hline
$CR^1_\alpha$ $\alpha\in[0,1)$ &\checkmark&IO&X \cite{SL17}&\checkmark &\\
\hline
$CR^2_\alpha$ $\alpha\in(0,2]$ &\checkmark&GIO &see (a)&&X\\
\hline
$CR^3_\alpha$ &\checkmark&DIO \cite{CG16-2}&&&\\
\hline
$CT^1_\alpha$ $\alpha\in[0,1)$&\checkmark&IO&X&\checkmark&\\
\hline
$CT^2_\alpha$  $\alpha\in(0,2]$&\checkmark&GIO &see (a)&&X\\
\hline
$C_f$ &\checkmark&GIO& see (a)&&X\\
\hline
\end{tabular}

\vspace{0.2in}


The fact that $CT^2_\alpha$ and $CR^2_\alpha$ are monotone under GIO can be derived from GIO monotonicity of $C_f$ \cite{V22}, or it was shown separately in \cite{DH21}. There are examples when the monotonicity of both are violated under a larger class of operators when $\alpha>1$,\cite{DH21}.

$CT^2_\alpha$ satisfies a modified version of additivity (C5), which $CR^2_\alpha$ also violates \cite{DH21},
$$CT^2_\alpha(p_1\rho_1\oplus p_2\rho_2)=p_1^\alpha CT^2_\alpha(\rho_1)+p_2^\alpha CT^2_\alpha(\rho_2) \ . $$

(a) In \cite{V22} it was shown that $C_f$, and in particular $CR^2_\alpha$ and $CT^2_\alpha$, reach equality in the strong monotonicity under a convex mixture of diagonal unitaries in any dimension, which implies these coherences reach equality in strong monotonicity under GIO in 2- and 3- dimensions. Moreover, these coherences are strongly monotone under GIO on pure states in any dimension.

$CR^1_\alpha(\rho)$, $CT^1_\alpha(\rho)$ violate strong monotonicity \cite{R16, SL17}. In \cite{R16} it was shown that $CT^1_\alpha(\rho)$ satisfies a modified version of the strong monotonicity: for $\alpha\in(0,2]$
$$\sum_n p_n^\alpha q_n^{1-\alpha}CT^1_\alpha(\rho_n)\leq CT^1_\alpha(\rho) \ ,$$
where $p_n=\Tr(K_n\rho K_n^*)$, $q_n=\Tr(K_n \Delta_\alpha(\rho)K_n^*)$ and $\rho_n$ is a post-measurement state.

Clearly, (C6) is not satisfied for any R\'enyi or Tsallis coherences in its original form, therefore it was not included in the list. However, the values of coherences on pure states can be easily calculated in some cases.

\section{$f$-coherence distillation}
\subsection{Continuity of $f$-entropy and $f$-coherence}

In addition to the above list of properties of the $f$-coherence, one can add its continuity in the following form (this is a direct application of result in  \cite{Pin21}).
\begin{lemma}\label{thm:entropy-cont}
Let $\rho$ and $\sigma$ be two states such that $\epsilon:=\frac{1}{2}\|\rho-\sigma\|_1$. Then
\begin{align*} 
|{S}^1_f(\rho)-{S}^1_f(\sigma)|&\leq -(1-\epsilon)f\left(\frac{1}{1-\epsilon}\right)-\epsilon f\left(\frac{d-1}{\epsilon} \right)\\
|{S}_f^2(\rho)-{S}_f^2(\sigma)|&\leq f\left(\frac{1}{d} \right)-(1-\epsilon)f\left(\frac{1}{d(1-\epsilon)}\right)-\epsilon f\left(\frac{d-1}{d\epsilon} \right)\ .
\end{align*}
 Denote either of the right hand-sides as ${H}(\epsilon)$, and note that ${H}$ is continuous in $\epsilon$, and goes to zero when $\epsilon\rightarrow 0$.

\end{lemma}

\begin{proof}
Recall that for any convex function $f$, the transpose of it $\tilde{f}(x)=xf(1/x)$ is also convex. We adapt a convention $0\cdot \infty=0$, so for a convex  function $f$ such that $f(1)=0$, we have $\tilde{f}(0)=\tilde{f}(1)=0$. Then $f$- entropy (\ref{eq:f-entropy-hat}) can be written using a transpose function as
$$ {S}^1_f(\rho)=-S_f(\rho\|I)=-\Tr(\rho f(\rho^{-1}))=-\Tr(\tilde{f}(\rho))\ ,$$
and
\begin{align*} 
&S_f^2(\rho)=-S_f(\rho\|I/d)=f(1/d)-\Tr(\rho f(\{d\rho\}^{-1}))\\
&=f(1/d)-\frac{1}{d}\Tr(\tilde{f}(d\rho))\ .
\end{align*}

In \cite{Pin21} Theorem 1, it was proved that for $S_f(\rho)=-\Tr g(\rho)$ and any convex function $g$ the following holds
$$|S_g(\rho)-S_g(\sigma)|\leq g(1)-g(1-\epsilon)-(d-1)\left(g\left(\frac{\epsilon}{d-1} \right)-g(0) \right)\ , $$
when $\epsilon=\frac{1}{2}\|\rho-\sigma\|_1$.
 And in Corollary 3, the result was generalized for non-unit trace density matrices: let $\rho$ and $\sigma$ be two states of the same trace $t$, and let $\epsilon=\frac{1}{2}\|\rho-\sigma\|_1\in[0,t]$, then
 $$|S_g(\rho)-S_g(\sigma)|\leq g(t)-g(t-\epsilon)-(d-1)\left(g\left(\frac{\epsilon}{d-1} \right)-g(0) \right)\ . $$

Adapting this result to our situation, it holds that
\begin{align*} 
|{S}^1_f(\rho)-{S}^1_f(\sigma)|&\leq -(1-\epsilon)f\left(\frac{1}{1-\epsilon}\right)-\epsilon f\left(\frac{d-1}{\epsilon} \right)\ .
\end{align*}
 And similarly, for $\tilde{\epsilon}:=d\epsilon=\frac{1}{2}\|d\rho-d\sigma\|_1\in[0,d]$
 \begin{align*} 
&|{S}_f^2(\rho)-{S}_f^2(\sigma)|\\
&=\frac{1}{d}\left|\Tr(\tilde{f}(d\rho))-\Tr(\tilde{f}(d\sigma))\right|\\
&\leq\frac{1}{d}\left[\tilde{f}(d)-\tilde{f}(d-\tilde{\epsilon})-(d-1)\left(\tilde{f}(\tilde{\epsilon}/(d-1))-\tilde{f}(0)\right) \right]\\
&= f\left(\frac{1}{d} \right)-(1-\epsilon)f\left(\frac{1}{d(1-\epsilon)}\right)-\epsilon f\left(\frac{d-1}{d\epsilon} \right)\ .
\end{align*}
\end{proof}

 From this continuity result, one can obtain continuity of the $f$-coherence. 
 \begin{theorem}\label{thm:coh-cont}
 Let $\rho$ and $\sigma$ be two states such that $\epsilon:=\frac{1}{2}\|\rho-\sigma\|_1$. Let $H(\epsilon)$ be as in the previous theorem for the corresponding $f$-entropy. Then for  $f$-coherences we obtain
  \begin{align*} 
|{C}_f(\rho)-{C}_f(\sigma)|\leq 2{H}(\epsilon)\ .
\end{align*}
\end{theorem}
 \begin{proof} 
Let $\rho$ and $\sigma$ be two states with $\epsilon=\frac{1}{2}\|\rho-\sigma\|_1$.  Since trace-norm is monotone under CPTP maps, in particular, under dephasing operation, it follows that
 $$\|\Delta(\rho)-\Delta(\sigma)\|_1\leq\|\rho-\sigma\|_1\leq 2\epsilon\ . $$
 Therefore, from continuity results above Theorem \ref{thm:entropy-cont}, for either $f$-coherence and the corresponding $f$-entropy, we obtain
 \begin{align*} 
&|{C}_f(\rho)-{C}_f(\sigma)|\\
&\leq  |{S}_f(\Delta(\rho))-{S}_f(\Delta(\sigma))|+ |{S}_f(\rho)-{S}_f(\sigma)|\\
&\leq 2{H}(\epsilon)\ .
\end{align*}
\end{proof}

\subsection{Coherence distillation}

In \cite{DS16} it was shown that it is not possible to distill a higher coherence state $\sigma$ from a lower coherence state $\rho$ via GI operations when coherence is measured by a relative entropy of coherence (which equal to the distillable coherence). The same result holds for $f$-coherences as well, which relies on the continuity property of coherence above, and the GIO monotonicity of $f$-coherence \cite{V22}. For completeness sake, we present the adapted proof from \cite{DS16} below.

\begin{definition}
A state $\sigma$ can be distilled from the state $\rho$ at rate $0 < R \leq 1$ if there exists an operation $\rho^{\otimes n}\rightarrow \tau$ such that $\|\Tr_{ref}\tau-\sigma^{\otimes nR}\|_1\leq \epsilon$ and $\epsilon\rightarrow 0$ as $n\rightarrow\infty$. The optimal
rate at which distillation is possible is the supremum of $R$ over all protocols fulfilling the aforementioned conditions.
\end{definition}

\begin{theorem}\label{thm:distill}
 Given two states $\rho$ and $\sigma$ such that
$$C_f(\rho)<C_f(\sigma)\ ,$$
it is not possible to distill $\sigma$ from $\rho$ at any rate $R>0$ via GIO operations.
\end{theorem}
\begin{proof} Suppose the contradiction holds, assume that there are two states $\rho$ and $\sigma$ such that $C_f(\rho)<C_f(\sigma)$, and that the distillation is possible. In particular, for large enough $n$,  it is possible to approximate one copy of $\sigma$. In other words, for any $\epsilon>0$, there is a GIO $\Lambda$ such that
$$\|\Tr_{n-1}\Lambda(\rho^{\otimes n})-\sigma\|_1\leq \epsilon\ . $$
By Lemma 12 in \cite{DS16}, there exists a GIO $\tilde{\Lambda}$ acting only on one copy of $\rho$, such that
$$\Tr_{n-1}\Lambda(\rho^{\otimes n})=\tilde{\Lambda}(\rho)\ . $$
Thus, for any $\epsilon>0$, there is a GIO $\tilde{\Lambda}$ such that
$$\|\tilde{\Lambda}(\rho)-\sigma\|_1\leq \epsilon\ . $$
Using the asymptotic continuity of $f$-coherence, Theorem \ref{thm:coh-cont}, for these two $\epsilon$-close states, we obtain
$$\left|C_f(\tilde{\Lambda}(\rho))-C_f(\sigma) \right| \leq 2H(\epsilon/2) \ .$$
Recall that $H(\epsilon)$ for either $f$-coherence is continuous in $\epsilon\in(0,1)$ and it goes to zero when $\epsilon\rightarrow 0$. Therefore, summarizing from the beginning, for any $\delta>0$, there is GIO $\tilde{\Lambda}$ such that
\begin{equation}\label{eq:delta}
\left|C_f(\tilde{\Lambda}(\rho))-C_f(\sigma) \right| < \delta \ .
\end{equation}
Take $\delta:=\frac{1}{2}(C_f(\sigma)-C_f(\rho))>0$. Since $C_f$ is GIO monotone, for any GIO $\Lambda$, we have $$C_f(\tilde{\Lambda}(\rho))\leq C_f(\rho)\ .$$
Therefore,
$$\delta\leq \frac{1}{2}(C_f(\sigma)-C_f(\tilde{\Lambda}(\rho))< C_f(\sigma)-C_f(\tilde{\Lambda}(\rho))\ .$$ 
This is a contradiction to (\ref{eq:delta}).

\end{proof}

\section{New R\'enyi and Tsallis coherences}

Playing off the last expression in the definition of the relative entropy of coherence \ref{RelativeCoh}, we define coherence measure as follows:
$$CT_\alpha(\rho):=S_\alpha^T(\Delta_\alpha(\rho))-S^T_\alpha(\rho)\ , $$
for Tsallis entropy, and
$$CR_\alpha(\rho):=S_\alpha^R(\Delta_\alpha(\rho))-S^R_\alpha(\rho)\ , $$
for R\'enyi entropy. Recall that here $\Delta_\alpha(\rho)$ is the closest incoherent state to $\rho$ when measured by the R\'{e}nyi or Tsallis relative entropy, i.e.
$$S_\alpha(\rho\|\Delta_\alpha(\rho)):=\min_{\delta\in \mathcal{I}}S_\alpha(\rho\|\delta)\ . $$
Recall from (\ref{eq:Delta-alpha}) that
 \begin{equation*}
\Delta_\alpha(\rho)=\frac{1}{N(\rho)}\sum_j\bra{j}\rho^\alpha\ket{j}^{1/\alpha}\ket{j}\bra{j}\ , 
\end{equation*}
where $N(\rho)=\sum_j\bra{j}\rho^\alpha\ket{j}^{1/\alpha}$. Having this explicit form of $\Delta_\alpha(\rho)$ , both coherences can be explicitly calculated
 \begin{align*} 
 CT_\alpha(\rho)&= \frac{1}{1-\alpha}\left[\Tr\left(\Delta_\alpha(\rho)^\alpha \right)-\Tr\rho^\alpha \right]\\
 &=\frac{1}{1-\alpha}\left[\frac{1}{N(\rho)^\alpha}-1\right]\Tr\rho^\alpha\\
 &=\frac{N(\rho)^\alpha-1}{\alpha-1}\frac{\Tr\rho^\alpha}{N(\rho)^\alpha}\\
 &=S_\alpha^T(\rho\|\Delta_\alpha(\rho))\, \frac{\Tr\rho^\alpha}{N(\rho)^\alpha}\\
 &=CT_\alpha^1(\rho)\, \frac{\Tr\rho^\alpha}{N(\rho)^\alpha}\\
 &\geq 0\ .
 \end{align*}
The last two equalities come from (\ref{eq:CT1}). Similarly, from (\ref{eq:CR1}) for the R\'enyi coherence
 \begin{align*} 
 CR_\alpha(\rho)&= \frac{1}{1-\alpha}\left[\log\Tr\left(\Delta_\alpha(\rho)^\alpha \right)-\log\Tr\rho^\alpha \right]\\
 &=\frac{1}{1-\alpha}\left[\log\left(\frac{1}{N(\rho)^\alpha}\Tr\rho^\alpha\right)-\log\Tr\rho^\alpha\right]\\
 &=\frac{\alpha}{\alpha-1}\log N(\rho)\\
 &=S_\alpha^R(\rho\|\Delta_\alpha(\rho))\\
 &=CR_\alpha^1(\rho)\ .
 \end{align*}
This means that for R\'enyi entropy of coherence we have a similar expressions to the relative entropy of coherence (\ref{RelativeCoh})
$$CR^1_\alpha(\rho)=\min_{\delta\in\mathcal{I}}S_\alpha^R(\rho\|\delta)=S_\alpha^R(\rho\|\Delta_\alpha(\rho))=S^R_\alpha(\Delta_\alpha(\rho))-S^R_\alpha(\rho)\ . $$
Therefore the distance-based R\'{e}nyi coherence $CR_\alpha^1(\rho)$ coincides with the new definition $CR_\alpha(\rho)$. Before moving on to investigation of the new Tsallis coherence, let us show one result on R\'enyi coherence.

\begin{theorem}\label{thm:cont-R}
Let $\rho=\ket{\psi}\bra{\psi}$ and $\sigma=\ket{\phi}\bra{\phi}$ be pure states on $\mathbb{C}^d$ such that $\frac{1}{2}\|\rho-\sigma\|_1=\epsilon$. Then, we obtain
$$\left|CR^1_\alpha(\rho)-CR^1_\alpha(\sigma)\right|\leq \frac{\alpha}{1-\alpha}\log \left(d^{1-\frac{1}{\alpha}}+H(\epsilon) \right) +\log d\ ,  \text{ for $0<\alpha<1$\ ,}$$
and 
$$\left|CR^1_\alpha(\rho)-CR^1_\alpha(\sigma)\right|\leq \frac{\alpha}{\alpha-1}\log (1-H(\epsilon)) \ ,  \text{ for $1<\alpha<2$\ ,}$$
where $H(\epsilon)=1-(1-\epsilon)^{1/\alpha}-\epsilon^{1/\alpha}(d-1)^{1-\frac{1}{\alpha}}$. Both right hand-sides converge to zero when $\epsilon$ goes to zero.
\end{theorem}
\begin{proof}
Denote $\chi_j=|\langle \psi| j\rangle|^2$ and  $\xi_j=|\langle \phi| j\rangle|^2$. Then,
\begin{align*}
|CR^1_\alpha(\rho)-CR^1_\alpha(\sigma)|&=\frac{\alpha}{|1-\alpha|}\left|\log\left(\sum_j\chi_j^{1/\alpha} \right)-\log\left(\sum_j\xi_j^{1/\alpha} \right) \right|\\
&=\frac{\alpha}{|1-\alpha|}\left|\log\Tr f(\Delta(\rho))-\log \Tr f(\Delta(\sigma)) \right|\ ,
\end{align*}
where $f(x)=x^{1/\alpha}$ is convex function for $0<\alpha<1$ and $-f$ is convex for $\alpha>1$, and  recall that $\Delta(\rho)=\sum_j\chi_j\ket{j}\bra{j}$ and $\Delta(\sigma)=\sum_j\xi_j\ket{j}\bra{j}$.

Since trace-norm is monotone under CPTP maps, and $\Delta$ is a CPTP map, we obtain
$$\frac{1}{2}\|\Delta(\rho)-\Delta(\sigma)\|_1\leq \frac{1}{2}\|\rho-\sigma\|_1=\epsilon\ . $$
 By continuity of $f$-entropy  \cite{Pin21}, the difference for $0<\alpha<1$ is bounded by
 $$\left|\Tr f(\Delta(\rho)) -\Tr f(\Delta(\sigma))  \right|\leq H(\epsilon)\ , $$
 where $H(\epsilon)$ is calculated for $f(x)=x^{1/\alpha}$, and therefore has expression as in the theorem statement.
 For $\alpha>1$, $-f$ is convex, and therefore, 
 $$\left|\Tr f(\Delta(\rho)) -\Tr f(\Delta(\sigma))  \right|\leq -H(\epsilon)\ , $$
 where the right-hand side is positive for $\alpha>1$.

For $0<\alpha<1$, notice that the constant sequence is majorized by both $(\frac{1}{d})_j\prec (\chi)_j$ and $(\frac{1}{d})_j\prec (\xi)_j$,   therefore, since $f(x)=x^{1/\alpha}$ is a convex function, by results on Schur-concavity \cite{HLP29, MOA10, S23}, we have $\sum_j\chi_j^{1/\alpha}, \sum_j\xi_j^{1/\alpha}\geq d^{1-\frac{1}{\alpha}}$. For $\alpha>1$,  since $x<x^{1/\alpha}$, then   $\sum_j\chi_j^{1/\alpha}, \sum_j\xi_j^{1/\alpha}>1$.

For the function $g(x)=\log x$, by the Mean Value Theorem, for $0<\alpha<1$, there exist $c\in[d^{1-\frac{1}{\alpha}}, 1]$, such that 
\begin{align*}
\left|\log\Tr f(\Delta(\rho)) -\log \Tr f(\Delta(\sigma)) \right|&=\left|\Tr f(\Delta(\rho)) -\Tr f(\Delta(\sigma))  \right||g'(c)|\ .
\end{align*}
Denote $\delta:=\left|\Tr f(\Delta(\rho)) -\Tr f(\Delta(\sigma))  \right|$. Then, by the Mean Value Theorem, there exists $s\in[d^{1-\frac{1}{\alpha}}, 1]$ such that $\left|g(d^{1-\frac{1}{\alpha}}+\delta)-g(d^{1-\frac{1}{\alpha}}) \right|=\delta\, g'(s)$. Since $s<c$, we have $|g'(s)|>|g'(c)|$, and therefore
\begin{align*}
\left|\log\Tr f(\Delta(\rho)) -\log \Tr f(\Delta(\sigma)) \right|&\leq \left|g(d^{1-\frac{1}{\alpha}}+\delta)-g(d^{1-\frac{1}{\alpha}}) \right|\\
&\leq \left|g(d^{1-\frac{1}{\alpha}}+H(\epsilon))-g(d^{1-\frac{1}{\alpha}}) \right|\\
&=\log \left(d^{1-\frac{1}{\alpha}}+H(\epsilon) \right) +\frac{1-\alpha}{\alpha}\log d\ .
\end{align*}

For $\alpha>1$, by the Mean Value Theorem, there exists $c\geq 1$, such that
\begin{align*}
\left|\log\Tr f(\Delta(\rho))-\log\Tr f(\Delta(\sigma))  \right|&=\left|\Tr f(\Delta(\rho)) -\Tr f(\Delta(\sigma))  \right||g'(c)|\ .
\end{align*}

Denote $\delta:=\left|\Tr f(\Delta(\rho)) -\Tr f(\Delta(\sigma))  \right|$. Then, by the Mean Value Theorem, there exists $s> 1$ such that $\left|g(1+\delta)-g(1) \right|=\delta\, g'(s)$. Since $s<c$, we have $|g'(s)|>|g'(c)|$, and therefore
\begin{align*}
\left|\log\Tr f(\Delta(\rho))-\log\Tr f(\Delta(\sigma))  \right|&\leq \left|g(1+\delta)-g(1) \right|\\
&\leq  \left|g(1-H(\epsilon))-g(1) \right|\\
&=\log (1-H(\epsilon))\ .
\end{align*}

Thus, we obtain the statement of the theorem.
\end{proof}

\section{Tsallis coherence}

\subsection{Positivity} As we noted above, the Tsallis coherence is non-negative. Note that this is a non-trivial statement, that cannot be directly observed by the monotonicity of entropy under linear CPTP maps, as it was done for $CT^2_\alpha, CR^2_\alpha, C_f$, since the map $\rho\rightarrow \Delta_\alpha(\rho)$ is non-linear.

\subsection{Vanishing only on incoherent states} 
\begin{proposition}
$CT_\alpha(\rho)=0$ if and only if $\rho\in\cI$ is incoherent.
\end{proposition}
\begin{proof}
First, suppose that the state $\rho\in\cI$ is incoherent, then $\Delta_\alpha(\rho)=\rho$. Therefore, $CT_\alpha(\rho)=S^T_\alpha(\Delta_\alpha(\rho))-S^T_\alpha(\rho)=0.$

Now, suppose that $CT_\alpha(\rho)=0.$ From calculations above, since $\Tr\rho^\alpha>0$ for a non-zero state, this means that $S_\alpha^T(\rho\|\Delta_\alpha(\rho))=0$, which happens only when $\rho=\Delta_\alpha(\rho)\in \cI$. Therefore, $\rho\in\cI$ is incoherent.
\end{proof}

\subsection{Value on pure states}
Let $\rho=\ket{\psi}\bra{\psi}$ be a pure state. Since $\rho^\alpha=\rho$, then
 \begin{align*} 
 CT_\alpha(\rho)&= \frac{1}{1-\alpha}\left[\Tr\left(\Delta_\alpha(\rho)^\alpha \right)-\Tr\rho^\alpha \right]=S_\alpha^T(\Delta_\alpha(\rho))\ .
 \end{align*}
To calculate this Tsallis entropy explicitly, we note that $\Tr(\Delta_\alpha(\ket{\psi}\bra{\psi})^\alpha)=N(\ket{\psi}\bra{\psi})^{-\alpha}$, where 
$N(\ket{\psi}\bra{\psi})=\sum_j |\langle\psi|j\rangle|^{2/\alpha} .$
Thus,
$$CT_\alpha(\rho)= \frac{1}{1-\alpha}\left[\left( \sum_j |\langle\psi|j\rangle|^{2/\alpha}\right)^{-\alpha}-1\right]\ .$$

\subsection{Comparison with $CT_\alpha^1$}
Recall that from our previous calculations,
 \begin{align*} 
 CT_\alpha(\rho)=CT_\alpha^1(\rho)\, \frac{\Tr\rho^\alpha}{N(\rho)^\alpha}\ .
 \end{align*}
 Let us denote as $\lambda_j:=\bra{j}\rho^\alpha\ket{j}$. Then 
 $$\Tr(\rho^\alpha)=\sum_j\bra{j}\rho^\alpha\ket{j}=\|\lambda\|_1\ .$$
 And 
 $$N(\rho)^\alpha=\left(\sum_j\bra{j}\rho^\alpha\ket{j}^{1/\alpha}\right)^\alpha=\|\lambda\|_{1/\alpha}\ .  $$
 Here $\|\cdot\|_p$ denotes the Schatten $p$-norm. Since Schatten $p$-norms are monotone decreasing in $p$, we have that
 $$CT_\alpha(\rho)\geq CT_\alpha^1(\rho)\ , \text{ for }0<\alpha<1\ , $$
 and 
  $$CT_\alpha(\rho)\leq CT_\alpha^1(\rho)\ , \text{ for }1<\alpha<2\ . $$

\subsection{Monotonicity} 
\begin{theorem}\label{thm:inv-uni}
$CT_\alpha(\rho)$ is invariant under diagonal unitaries. 
\end{theorem}
\begin{proof}
Let $U=\sum_n e^{i\phi_n}\ket{n}\bra{n}$ be a unitary diagonal in  $\cE$ basis. Then 
\begin{align*}
\Delta_\alpha(U\rho U^*)&=\frac{1}{\sum\bra{j}U\rho^\alpha U^*\ket{j}^{1/\alpha}}\sum\bra{j}U\rho^\alpha U^*\ket{j}^{1/\alpha}\ket{j}\bra{j}\\
&=\frac{1}{\sum\bra{j}e^{i\phi_j}\rho^\alpha e^{-i\phi_j} \ket{j}^{1/\alpha}}\sum\bra{j}e^{i\phi_j}\rho^\alpha e^{-i\phi_j}\ket{j}^{1/\alpha}\ket{j}\bra{j}\\
&=\Delta_\alpha(\rho)
\end{align*}
Since the Tsallis entropy is invariant under unitaries itself, we have
$$CT_\alpha(U\rho U^*)=CT_\alpha(\rho)\ . $$
\end{proof}

\begin{theorem}\label{thm:monoGIO}
Tsallis coherence is not monotone under GIO.
\end{theorem}
\begin{proof}
Let us fix the basis $\cE=\{\ket{0},\ket{1}\}$ Let $\rho=\ket{\psi}\bra{\psi}$ be a pure state with $|\langle \psi | 0 \rangle|^2=\chi=3/4$ and  $|\langle \psi | 1 \rangle|^2=1-\chi=1/4$.

 For a pure state $\rho$ the entropy is zero, and therefore
 \begin{align*}
 CT_\alpha(\rho)&=S^T_\alpha(\Delta_\alpha(\rho))-S^T_\alpha(\rho)\\
 &=S^T_\alpha(\Delta_\alpha(\rho)) \\
 &=\frac{1}{1-\alpha}\left[\Tr\left\{\Delta_\alpha(\rho))^\alpha \right\}-1  \right]\\
  &=\frac{1}{1-\alpha}\left[\frac{1}{\left(\sum_j\chi_j^{1/\alpha}\right)^\alpha}-1  \right]\\
   &=\frac{1}{1-\alpha}\left[\frac{4}{\left(3^{1/\alpha}+1\right)^\alpha}-1  \right]\ . 
  \end{align*}
  Let $\Lambda$ be GIO,  with Kraus operators $\Lambda(\rho)=K_1 \rho K_1^*+K_2 \rho K_2^*$ where Kraus operators are diagonal in $\cE$ basis
 $$K_1=\begin{pmatrix}
\frac{1}{\sqrt{2}} & 0 \\
0 & \frac{\sqrt{3}}{2} 
\end{pmatrix}\ , \qquad
K_2=\begin{pmatrix}
\frac{1}{\sqrt{2}} & 0 \\
0 & \frac{1}{2}
\end{pmatrix}\ .$$
Clearly $\sum_n K_n^*K_n=I$. Then 
$$\Lambda(\rho)=\begin{pmatrix}
\frac{3}{4} & a \\
a& \frac{1}{4}
\end{pmatrix}\ , $$
where $a=\frac{3+\sqrt{3}}{8\sqrt{2}}$. The eigenvalues of this matrix are
$\beta_{1,2}=\frac{1}{2}\left(1\pm\sqrt{\frac{1}{4}+4a^2}\right)$. And the normalized eigenvector corresponding to $\beta_{1,2}$ are
$$\ket{\psi_{1,2}}=\frac{1}{\sqrt{a^2+(\beta_{1,2}-\frac{3}{4})^2}}\begin{pmatrix}
a  \\
\beta_{1,2}-\frac{3}{4}
\end{pmatrix}\ . $$
Therefore, $\Tr(\Lambda(\rho)^\alpha)=\beta_1^\alpha+\beta_2^\alpha$, and
$$N(\Lambda(\rho))=\sum_j \beta_1|\langle j|\psi_1\rangle|^{2/\alpha}+\beta_2|\langle j|\psi_2\rangle|^{2/\alpha}\ . $$
And the Tsallis coherence is then
$$CT_\alpha(\Lambda(\rho))=\frac{1}{1-\alpha}\left[\frac{1}{N(\Lambda(\rho))^\alpha}-1 \right]\Tr(\Lambda(\rho)^\alpha)\ . $$

\begin{figure}[t!]
\begin{center}
  \includegraphics[width=\linewidth/2]{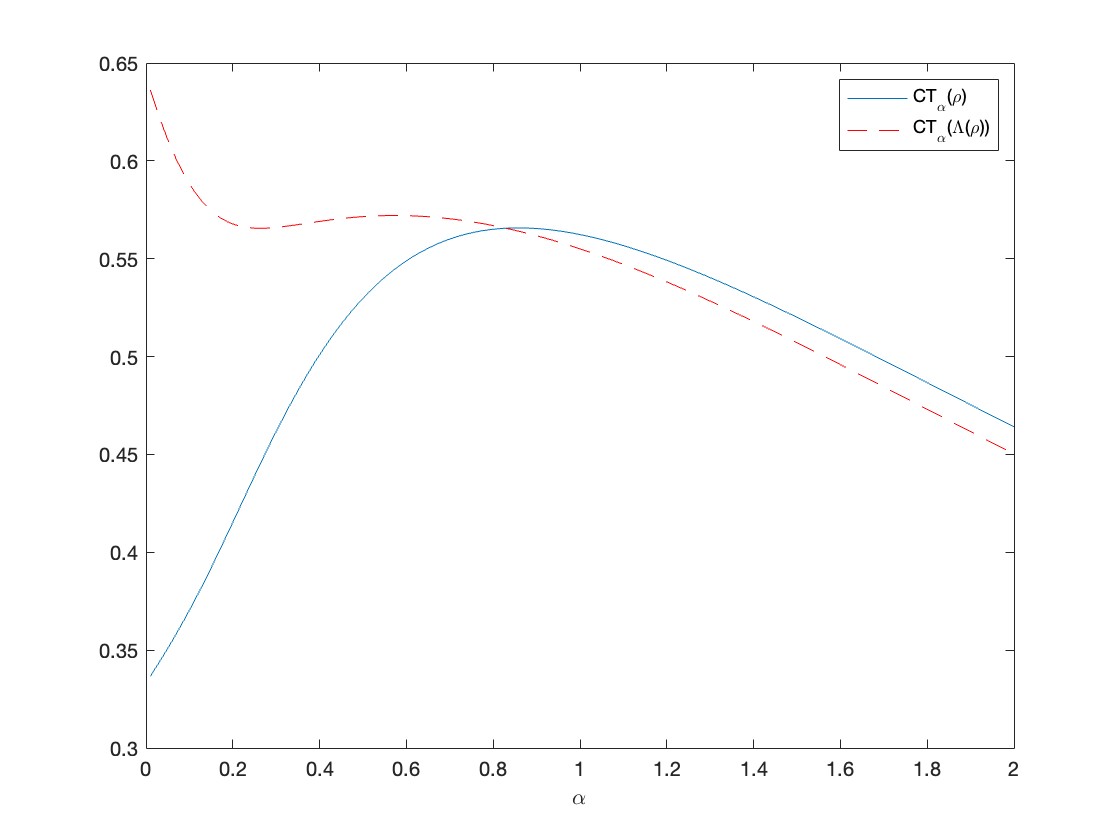}
  \caption{Failure of monotonicity under GIO for small $\alpha$.}
  \label{fig:MonoGIO}
  \end{center}
\end{figure}

From Figure \ref{fig:MonoGIO}, we see that, for example, for $\alpha=0.2$, monotonicity has failed 
$$CT_\alpha(\rho)<0.5<CT_\alpha(\Lambda(\rho))\ . $$
\end{proof}

\begin{definition}
A GIO map $\Lambda$ that commutes with $\Delta_\alpha$ is called {\bf $\alpha$-GIO}. 
\end{definition}
A unitary diagonal under a fixed basis $\cE$ is an $\alpha$-GIO for any $\alpha$. For $\alpha=1$, $\Delta_\alpha(\rho)=\Delta(\rho)$, which commutes with any GIO.

\begin{theorem}
Tsallis coherence is monotone under $\alpha$-GIO.
\end{theorem}
\begin{proof}
By definition 
$$CT_\alpha(\rho)-CT_\alpha(\Lambda(\rho))=S_\alpha^T(\Lambda(\rho))-S_\alpha^T(\rho)+S_\alpha^T(\Delta_\alpha(\rho))-S_\alpha^T(\Delta_\alpha(\Lambda(\rho)))\ .$$
Since Tsallis entropy is monotone under CPTP maps, $S_\alpha^T(\Lambda(\rho))-S_\alpha^T(\rho)\geq 0$.
$\Lambda$ commutes with $\Delta_\alpha$,  and $\Lambda$ is GIO, so it leaves the incoherent states, such as $\Delta_\alpha(\rho)$, invariant, therefore 
\begin{align*}
S_\alpha^T(\Delta_\alpha(\rho))-S_\alpha^T(\Delta_\alpha(\Lambda(\rho)))&=S_\alpha^T(\Delta_\alpha(\rho))-S_\alpha^T(\Lambda(\Delta_\alpha(\rho)))=0 \ .
\end{align*}
\end{proof}

\subsection{Strong monotonicity.} 
\begin{theorem}
Tsallis coherence $CT_\alpha(\rho)$ reaches equality in strong monotonicity for convex mixtures of diagonal unitaries. Therefore, $CT_\alpha(\rho)$ reaches equality in strong monotonicity under GIO in two- and three-dimensions, when Kraus operators are proportional to diagonal unitaries.
\end{theorem}
\begin{proof}
Consider a GIO $\Lambda$ that is a probabilistic mixture of diagonal unitaries, i.e. let  
$$\Lambda(\rho)=\sum_k \alpha_k U_k\rho U_k^*\ ,$$ 
where $\alpha_j\in[0,1]$ with $\sum \alpha_k=1$, and the unitaries $U_k$ are diagonal in $\cE$. Then from Theorem \ref{thm:inv-uni}, since $CT_\alpha$ is invariant under diagonal unitaries, we have
\begin{align*}
\sum_k \alpha_kCT_\alpha(U_k\rho U_k^*)=\left(\sum_k \alpha_k \right)CT_\alpha(\rho)=CT_\alpha(\rho)\ .
\end{align*}

\end{proof}

In general, $CT_\alpha$ fails strong monotonicity for IO maps.
\begin{theorem}
Tsallis coherence $CT_\alpha(\rho)$ fails strong monotonicity under IO maps.
\end{theorem}
\begin{proof}
We use example from \cite{SL17}, which was used to show that $CR^1_\alpha$ fails strong monotonicity under IO maps. Consider a three-dimensional space spanned by standard orthonormal basis $\cE=\{\ket{0}, \ket{1}, \ket{2}\}$. Let the density matrix be
$$\rho=\frac{1}{4}\begin{pmatrix}
1 & 0 & 1\\
0 & 2 & 0\\
1 & 0 & 1\\
\end{pmatrix}\ .$$
Let the Kraus operators of the IO map be 
$$K_1=\begin{pmatrix}
0 & 1 & 0\\
0 & 0 & 0\\
0 & 0 & a\\
\end{pmatrix}\ , \qquad
K_2=\begin{pmatrix}
1 & 0 & 0\\
0 & 0 & b\\
0 & 0 & 0\\
\end{pmatrix}\ .
$$
Here $|a|^2+|b|^2=1$ to satisfy the condition $K_1^*K_1+K_2^*K_2=I$.
It is straightforward to check that these Kraus operators leave the space of incoherent states $\cI$ invariant.
The output states are
$$ \rho_1=\frac{1}{p_1}K_1\rho K_1^*=\frac{1}{2+|a|^2}\begin{pmatrix}
2 & 0 & 0\\
0 & 0 & 0\\
0 & 0 & |a|^2\\
\end{pmatrix}\ , \qquad
\rho_2=\frac{1}{p_2}K_2\rho K_2^*=\frac{1}{1+|b|^2}\begin{pmatrix}
1 & b^* & 0\\
b & |b|^2 & 0\\
0 & 0 & 0\\
\end{pmatrix}\ ,$$
where $p_1=\frac{2+|a|^2}{4}$ and $p_2=\frac{1+|b|^2}{4}$.
Notice that $\rho_1\in\cI$ is diagonal and therefore incoherent, and $\rho_2=\ket{\psi}\bra{\psi}$ is the pure state with $\ket{\psi}=\frac{1}{\sqrt{1+|b|^2}}(\ket{0}+b\ket{1})$.

The $\alpha$ power of $\rho$ is the state
$$\rho^\alpha=\frac{1}{2^{1+\alpha}}\begin{pmatrix}
1 & 0 & 1\\
0 & 2 & 0\\
1 & 0 & 1\\
\end{pmatrix}\ .$$
And therefore the Tsallis coherence is
$$CT_\alpha(\rho)=S_\alpha^T(\Delta_\alpha(\rho))-S_\alpha^T(\rho)=\frac{4}{1-\alpha}\left[(2+2^{1/\alpha})^{-\alpha}-2^{-(1+\alpha)} \right]\ . $$
Since $\rho\in\cI$ is incoherent, $CT_\alpha(\rho_1)=0$. And since $\rho_2$ is a pure state, the Tsallis coherence is
$$p_2CT_\alpha(\rho_2)=p_2S_\alpha^T(\Delta_\alpha(\rho_2))=\frac{1}{1-\alpha}\frac{1+|b|^2}{4}\left[(1+|b|^2)(1+|b|^{2/\alpha})^{-\alpha}-1 \right]\ . $$

\begin{figure}[t!]
\begin{center}
  \includegraphics[width=\linewidth/2]{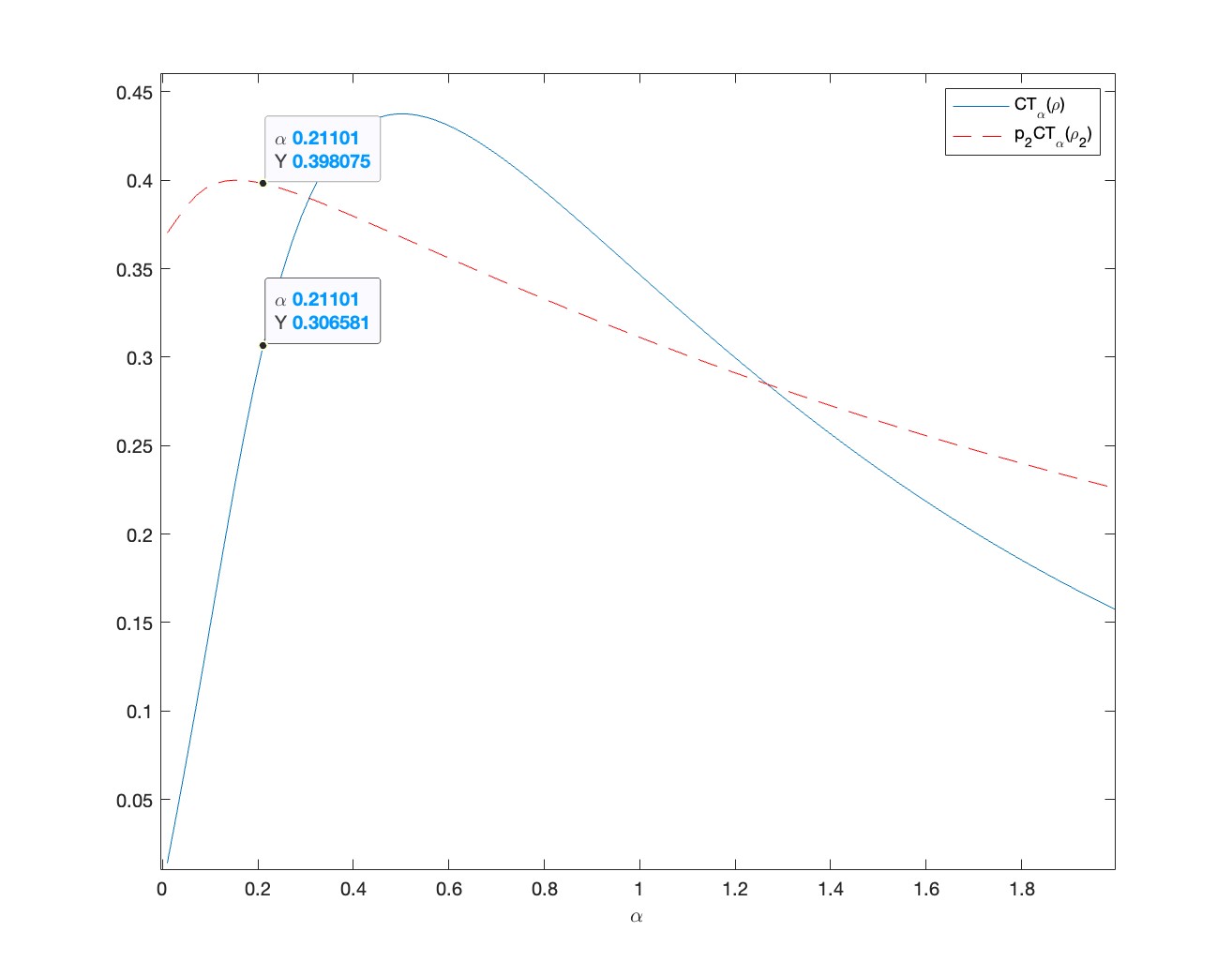}
  \caption{Failure of strong monotonicity under IO.}
  \label{fig:StrongMonoIO}
  \end{center}
\end{figure}

From Figure \ref{fig:StrongMonoIO} we have, for example, for $b=0.9$ and $\alpha=0.21101$, we have $$CT_\alpha(\rho)<0.35<p_2CT_\alpha(\rho_2)=\sum_j p_j CT_\alpha(\rho_j)\ .$$ 
\end{proof}

For strong monotonicity property it is important how the quantum channel is written in terms of its Kraus operators. We showed that in 2- or 3-dimensions, if GIO is written as a convex mixture of diagonal unitaries, then Tsallis coherence reaches equality. However, if GIO is written in some other way, we show that Tsallis coherence may fail strong monotonicity.
\begin{theorem}\label{thm:failed-strong}
Tsallis coherence fails strong monotonicity under GIO, even on pure states, if Kraus operators are not  proportional to unitaries.
\end{theorem}
\begin{proof}
We are going to use the same example as in Theorem \ref{thm:monoGIO}. Let us fix the basis $\cE=\{\ket{0},\ket{1}\}$ Let $\rho=\ket{\psi}\bra{\psi}$ be a pure state with $|\langle \psi | 0 \rangle|^2=\chi=3/4$ and  $|\langle \psi | 1 \rangle|^2=1-\chi=1/4$.

 For a pure state $\rho$ the entropy is zero, and therefore
 \begin{align*}
 CT_\alpha(\rho)&=S^T_\alpha(\Delta_\alpha(\rho)) \\
 &=\frac{1}{1-\alpha}\left[\Tr\left\{\Delta_\alpha(\rho))^\alpha \right\}-1  \right]\\
  &=\frac{1}{1-\alpha}\left[\frac{1}{\left(\sum_j\chi_j^{1/\alpha}\right)^\alpha}-1  \right]\\
   &=\frac{1}{1-\alpha}\left[\frac{4}{\left(3^{1/\alpha}+1\right)^\alpha}-1  \right]\ . 
  \end{align*}

 Let $\Lambda$ be GIO,  with Kraus operators $\Lambda(\rho)=K_1 \rho K_1^*+K_2 \rho K_2^*$ where Kraus operators are diagonal in $\cE$ basis
 $$K_1=\begin{pmatrix}
\frac{1}{\sqrt{2}} & 0 \\
0 & \frac{\sqrt{3}}{2} 
\end{pmatrix}\ , \qquad
K_2=\begin{pmatrix}
\frac{1}{\sqrt{2}} & 0 \\
0 & \frac{1}{2}
\end{pmatrix}\ .$$
Clearly $\sum_n K_n^*K_n=I$. Then the post-measurement states $\rho_n=\frac{1}{p_n}K_n\rho K_n^*=\ket{\psi_n}\bra{\psi_n}$ are also pure, where $\ket{\psi_n}=\frac{1}{\sqrt{p_n}}K_n\ket{\psi}$ and $p_n=\bra{\psi}K_n^*K_n\ket{\psi}$.  Let us denote  $|\langle \psi_n | j \rangle|^2=\xi_{nj}=\frac{1}{p_n}|\bra{j}K_n\ket{\psi}|^2=\frac{1}{p_n}|k_{nj}|^2\chi_j$, and $p_n=\sum_j |k_{nj}|^2\chi_j$. Then $p_1=\frac{9}{16}$ and $p_2=\frac{7}{16}$, and 
$$\xi_{11}=\frac{2}{3}\ , \  \xi_{12}=\frac{1}{3}\, \qquad \xi_{21}=\frac{6}{7}\ , \ \xi_{22}=\frac{1}{7}\ .$$
Therefore,
 \begin{align*}
  CT_\alpha(\rho_1)&=S^T_\alpha(\Delta_\alpha(\rho_1))\\
 &=\frac{1}{1-\alpha}\left[\Tr\left\{\Delta_\alpha(\rho_1))^\alpha \right\}-1  \right]\\
  &=\frac{1}{1-\alpha}\left[\frac{1}{\left(\sum_j\xi_{1j}^{1/\alpha}\right)^\alpha}-1  \right]\\
  &=\frac{1}{1-\alpha}\left[\frac{3}{\left(2^{1/\alpha}+1\right)^\alpha}-1  \right]\ .\\
 \end{align*}
 Similarly, 
\begin{align*}
  CT_\alpha(\rho_2)&=S^T_\alpha(\Delta_\alpha(\rho_2))\\
 &=\frac{1}{1-\alpha}\left[\Tr\left\{\Delta_\alpha(\rho_2))^\alpha \right\}-1  \right]\\
  &=\frac{1}{1-\alpha}\left[\frac{1}{\left(\sum_j\xi_{2j}^{1/\alpha}\right)^\alpha}-1  \right]\\
  &=\frac{1}{1-\alpha}\left[\frac{7}{\left(6^{1/\alpha}+1\right)^\alpha}-1  \right]\ .\\
 \end{align*}
 
\begin{figure}[t!]
\begin{center}
  \includegraphics[width=\linewidth/2]{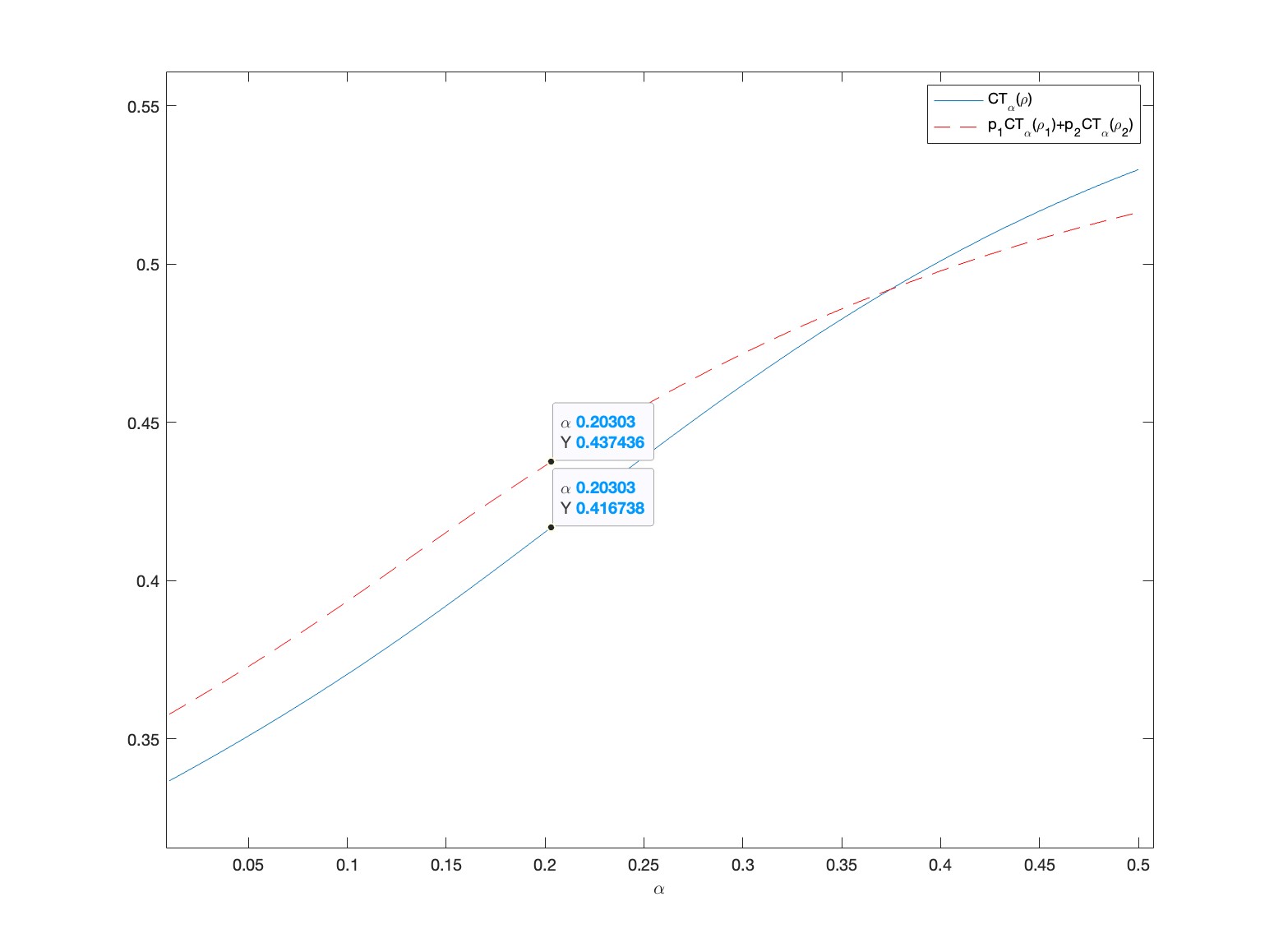}
  \caption{Failure of strong monotonicity under GIO for small $\alpha$.}
  \label{fig:StrongMonoGIO}
  \end{center}
\end{figure}

From Figure \ref{fig:StrongMonoGIO} we have, for  example, for $\alpha=0.20303$, strong monotonicity fails since
$$CT_\alpha(\rho)<0.42<p_1CT_\alpha(\rho_1)+p_2CT_\alpha(\rho_2) \ .$$
\end{proof}

\section{Improved $\alpha$-coherence measure}

Note that even though $\Delta_1=\Delta$, these two operators scale differently, in the following sense: $\Delta(p\rho)=p\Delta(\rho)$, and $\Delta_\alpha(p\rho)=\Delta(\rho)$. For this reason, define the ``unnormalized" $\Delta_\alpha$,
\begin{equation}\label{eq-delta-unnorm}
\tilde{\Delta}_\alpha(\rho)=\sum_j \bra{j}\rho^\alpha\ket{j}^{1/\alpha}\ket{j}\bra{j}\ .
\end{equation}
Note that $\tilde\Delta_\alpha(\rho)=\Delta(\rho^\alpha)^{1/\alpha}$.

 In \cite{DH21}, a coherence measure was proposed
 \begin{equation}\label{eq-impr-1}
\Tr\left|\Delta(\rho)^\alpha-\rho^\alpha \right|^{1/\alpha} \ ,
\end{equation}
 which was shown to be satisfy (C5). Since (C5) is equivalent to (C3) and (C4), and the later two imply (C2), satisfying (C5) implies that the expression is a coherence measure.

 Similarly to this, we propose the following coherence measures
 \begin{equation}\label{eq-impr-1}
C^1_\alpha(\rho)=\Tr\left|\tilde\Delta_\alpha(\rho)-\rho \right|=\Tr\left|\Delta(\rho^\alpha)^{1/\alpha}-\rho \right| \ ,
\end{equation}
and
\begin{equation}\label{eq-impr-1}
C^2_\alpha(\rho)=\Tr\left|\tilde\Delta_\alpha(\rho)^\alpha-\rho^\alpha \right|^\frac{1}{\alpha} =\Tr\left|\Delta(\rho^\alpha)-\rho^\alpha \right|^\frac{1}{\alpha} \ .
\end{equation}
Both, $C_\alpha^1$ and $C_\alpha^2$, can be easily shown to satisfy (C5): for $p_1+p_2=1$, $p_1, p_2\geq 0$ and any two states $\rho_1$ and $\rho_2$,  $$\cC(p_1\rho_1\oplus p_2\rho_2)=p_1\cC(\rho_1)+p_2\cC(\rho_2)\ .$$





\vspace{0.3in}
\textbf{Acknowledgments.}  A. V. is supported by NSF grant DMS-2105583.
\vspace{0.3in}

Data sharing not applicable to this article as no datasets were generated or analyzed during the current study.

The author has no competing interests or conflict of interest to declare that are relevant to the content of this article.

\end{document}